\documentclass[aps,prd,preprint,superscriptaddress,longbibliography,nofootinbib,floatfix]{revtex4-2}
\usepackage{CJKutf8}
\usepackage{booktabs}
\usepackage{graphicx}
\usepackage{amsmath,amssymb,amsthm}
\usepackage{mathrsfs}
\usepackage{orcidlink}
\numberwithin{equation}{section}

\newtheorem{theorem}{Theorem}[section]

\theoremstyle{definition}

\theoremstyle{remark}

\begin{document}

\title{Attraction and repulsion along extremal Einstein--two-Maxwell--dilaton branches}

\author{Ye Zhou (\begin{CJK*}{UTF8}{gbsn}周烨\end{CJK*})\orcidlink{0009-0004-7050-7736}}
\email{ye.zhou.horizon@gmail.com}
\affiliation{Independent Researcher, Kunshan, Jiangsu, China}

\begin{abstract}
We study extremal mass branches in four-dimensional Einstein gravity coupled to one dilaton and two Maxwell fields with opposite-sign dilaton couplings. For generic couplings the extremal mass is governed by a nonlinear first-order equation with no known closed-form solution. We prove that this equation has a unique positive $C^2$ solution on the full real charge-ratio line and that the solution is real analytic, so positivity on the entire positive charge cone selects a unique global mass branch. Along this branch, any two nonproportional positive charge vectors attract for $ab>-1$, have vanishing leading force at $ab=-1$, and repel for $ab<-1$. The same transition fixes the transverse convexity of the extremal mass surface and the sign of the mass difference under arbitrary finite positive-charge partitions. The sign theorem extends to higher dimensions through the known dimensional rescaling. A radial reconstruction then shows that every point of the mass branch defines an asymptotically flat exterior with a finite-area inner endpoint, while smooth or analytic horizon extension is a separate condition that restricts the couplings and scalar flow.
\end{abstract}

\maketitle

\section{Introduction}\label{Sec:Intro}

Extremal charged black holes provide a natural setting in which mass, charge, scalar hair, and long-range interactions can be compared. In Einstein--Maxwell theory, the balance between gravitational attraction and electrostatic repulsion underlies the Majumdar--Papapetrou family \cite{Majumdar1947,Papapetrou1947}. Once massless scalars are present, however, extremality of an individual state no longer fixes the force between states with different charge ratios. This distinction also enters force-based formulations of weak-gravity ideas in theories with moduli \cite{RepulsiveForces2019,Heidenreich2020}.

Dilaton gravity has long served as a useful laboratory for these questions \cite{GibbonsMaeda1988,GarfinkleHorowitzStrominger1991}. More recent work has focused on interactions between inequivalent extremal states. Repulsive behavior occurs in two-vector dilatonic systems and in exact Toda examples \cite{BlackHolesRepel2019,TodaForce2022}. In Einstein--Maxwell--dilaton theory, the force between different charge ratios was found to track both the mass difference under charge addition and the convexity or concavity of the extremal mass surface \cite{EMD2023}. The correspondence is not universal: in Einstein--Maxwell--dilaton--axion theory, force and binding energy need not have the same sign \cite{EMDA2025}.

A particularly useful intermediate theory contains one dilaton and two independently coupled Maxwell fields. Exact black-hole solutions are known for special choices of the couplings, including Toda-integrable cases \cite{Lu2013,Abishev2017,Rank2Toda2026}, whereas generic couplings do not admit comparably simple closed forms. This has motivated solution-independent approaches to black-hole mass, scalar charge, and thermodynamics \cite{LuYangLu2025,YangLuLu2026}. The recent E2MD analysis of Ref.~\cite{E2MD2026} reduced the extremal mass problem to a nonlinear first-order equation, identified $ab=-1$ as the four-dimensional force-cancellation locus, and established the attraction--repulsion pattern in several analytically controlled regimes and integrable examples. What remained was to determine whether this sign change is a property of the entire positive extremal mass branch at arbitrary finite charge ratio, and whether the same global transition controls the curvature and charge-addition energetics of the mass surface.

We prove that it is. For $a>0>b$ and positive electric charges, positivity on the full open charge cone selects a unique global real-analytic mass branch. On this branch the force between any two nonproportional positive charge vectors has sign $-\operatorname{sgn}(1+ab)$, and the same transition controls the transverse curvature of the mass surface and the mass difference under charge addition. Because the mass equation contains the derivative quadratically, constructing the global branch is part of the theorem rather than a choice of sign convention. The result extends directly to higher dimensions through the known dimensional rescaling.

A separate issue is whether every point of this mass branch represents a smooth extremal black hole. Horizon regularity in dilatonic and more general extremal systems can impose additional restrictions \cite{Poletti1995,Triangular2015,SmoothHorizons2025,AxionSingular2026}. We therefore reconstruct the four-dimensional exterior without assuming horizon smoothness. The geometry is regular for $r>0$ and approaches a finite-area inner endpoint, but generic nonconstant-scalar flows need not extend through that endpoint smoothly. Thus the global mass branch is the correct object for the interaction theorem, while smooth-horizon black holes form a more restricted subset.

The paper is organized as follows. Section~\ref{Sec:MassBranch} establishes the global positive mass branch. Section~\ref{Sec:Interactions} determines the force sign and its relation to mass convexity and charge-combination energetics. Section~\ref{Sec:Extensions} discusses the higher-dimensional extension and the four-dimensional exterior geometry. Section~\ref{Sec:Conclusions} concludes.

\section{The extremal mass branch}\label{Sec:MassBranch}

We work in units $G=1$ and consider the four-dimensional Einstein--two-Maxwell--dilaton theory
\begin{equation}\label{eq:action}
S=\frac{1}{16\pi}\int \mathrm{d}^{4}x\,\sqrt{-g}\left[R-2(\partial\phi)^2-e^{-2a\phi}F_{\mu\nu}F^{\mu\nu}-e^{-2b\phi}\widetilde F_{\mu\nu}\widetilde F^{\mu\nu}\right].
\end{equation}
Related two-vector dilatonic black-hole models have been studied in a variety of settings and conventions \cite{Lu2013,Abishev2017}. We adopt the asymptotically flat normalization of Ref.~\cite{E2MD2026} and restrict to
\begin{equation}\label{eq:coupling-domain}
a>0>b,
\end{equation}
with two strictly positive electric charges $Q$ and $P$. Here $P$ denotes the charge carried by $\widetilde F$ and is not a magnetic charge. Up to interchanging the two gauge fields and reversing the sign of $\phi$, Eq.~\eqref{eq:coupling-domain} fixes the convention for opposite-sign couplings. Throughout this section we consider the static, spherically symmetric, asymptotically flat extremal sector of the two-derivative theory; rotation, a scalar potential, and higher-derivative corrections lie outside the present scope. We choose the reference asymptotic modulus to be $\phi_\infty=0$ and define the ADM mass and scalar charge through
\begin{equation}\label{eq:asymptotic-charges}
g_{tt}=-1+\frac{2M}{r}+O(r^{-2}),\qquad \phi=\frac{\Sigma}{r}+O(r^{-2}).
\end{equation}
This fixes our sign convention for $\Sigma$.

The reduction to a one-dimensional mass equation follows from three structural relations. First, homogeneity of the extremal mass gives
\begin{equation}\label{eq:homogeneity}
M(\lambda Q,\lambda P)=\lambda M(Q,P),\qquad \lambda>0.
\end{equation}
Second, the action is invariant under a constant shift of the dilaton accompanied by
\begin{equation}\label{eq:shift-symmetry}
\phi\rightarrow\phi+\phi_{0},\qquad F\rightarrow e^{a\phi_{0}}F,\qquad \widetilde F\rightarrow e^{b\phi_{0}}\widetilde F,
\end{equation}
under which the conserved charges transform as $Q\rightarrow e^{-a\phi_{0}}Q$ and $P\rightarrow e^{-b\phi_{0}}P$. Following Ref.~\cite{E2MD2026}, we use this symmetry to relate the scalar charge defined in Eq.~\eqref{eq:asymptotic-charges} to derivatives of the extremal mass along the induced charge-rescaling orbit. After setting $\phi_\infty=0$, this gives
\begin{equation}\label{eq:scalar-charge-derivative}
\Sigma=-aQ\,M_Q-bP\,M_P,
\end{equation}
where subscripts on $M$ denote partial derivatives evaluated at $\phi_\infty=0$. Because the dilaton-shift orbit rescales $Q$ and $P$, Eq.~\eqref{eq:scalar-charge-derivative} should not be identified with a variation of $\phi_\infty$ at fixed conserved charges; the latter is the thermodynamic variation entering the first-law discussion of Ref.~\cite{ScalarFirstLaw1996}. Finally, the radial Hamiltonian constraint on the static extremal branch yields \cite{E2MD2026}
\begin{equation}\label{eq:extremality-constraint}
M^{2}+\Sigma^{2}=Q^{2}+P^{2}.
\end{equation}
Equation~\eqref{eq:extremality-constraint} is an extremality relation, not an assumption about the force between two distinct solutions. In Sec.~\ref{Sec:Regularity} we will reconstruct the corresponding exterior directly and recover the same asymptotic charges from the radial flow.

It is convenient to use the charge ratio as the single nontrivial variable. Equation~\eqref{eq:homogeneity} implies $M=P f(q)$ with $q=Q/P$. Defining
\begin{equation}\label{eq:mass-parametrization}
\delta=a-b>0,\qquad q=e^{\delta x},\qquad f(q)=e^{-bx}h(x),
\end{equation}
one finds from Eq.~\eqref{eq:scalar-charge-derivative}
\begin{equation}\label{eq:mass-charge-relations}
Q=P e^{\delta x},\qquad M=P e^{-bx}h(x),\qquad \Sigma=-P e^{-bx}h'(x).
\end{equation}
Substitution into Eq.~\eqref{eq:extremality-constraint} gives the E2MD mass equation
\begin{equation}\label{eq:mass-equation}
h'(x)^{2}+h(x)^{2}=e^{2ax}+e^{2bx},
\end{equation}
which was derived in Ref.~\cite{E2MD2026}. Closely related solution-independent mass--charge relations for Einstein gravity coupled to two Maxwell fields and one dilaton were developed in Ref.~\cite{LuYangLu2025}, and this strategy has recently been extended to extract black-hole thermodynamics without first constructing the corresponding closed-form solutions \cite{YangLuLu2026}. For the present problem, however, Eq.~\eqref{eq:mass-equation} leaves one issue unresolved: because the derivative enters quadratically, the equation does not by itself select a global positive mass branch.

All finite positive charge ratios correspond to $x\in\mathbb{R}$, and the prefactor $P e^{-bx}$ in Eq.~\eqref{eq:mass-charge-relations} is positive. Hence a mass branch that is positive over the entire open charge cone must be represented by a function $h$ that is positive on the full real axis. Introduce
\begin{equation}\label{eq:R-def}
\mathcal R(x)=\sqrt{e^{2ax}+e^{2bx}},\qquad \kappa=\sqrt{-ab},\qquad c=\frac{\sqrt{1-8ab}-1}{2},
\end{equation}
so that $c>0$ and
\begin{equation}\label{eq:c-relation}
c(c+1)=2\kappa^{2}.
\end{equation}
The function $\mathcal R$ diverges at both ends of the real axis and has a unique minimum at
\begin{equation}\label{eq:center}
x_{0}=\frac{1}{2\delta}\log\!\left(-\frac{b}{a}\right),\qquad q_{0}=e^{\delta x_0}=\sqrt{-\frac{b}{a}},\qquad \mathcal R_{0}=\mathcal R(x_{0}).
\end{equation}
The next result fixes the branch used throughout the remainder of the paper.

\begin{theorem}[Global positive mass branch]\label{thm:global-branch}
Equation~\eqref{eq:mass-equation} has exactly one solution $h\in C^{2}(\mathbb{R})$ satisfying $h(x)>0$ for all $x\in\mathbb{R}$. This solution is real analytic and obeys
\begin{equation}\label{eq:center-data}
h(x_{0})=\mathcal R_{0},\qquad h'(x_{0})=0,\qquad h''(x_{0})=c\mathcal R_{0}.
\end{equation}
Moreover,
\begin{equation}\label{eq:branch-monotonicity}
h'(x)<0\quad (x<x_{0}),\qquad h'(x)>0\quad (x>x_{0}),\qquad h(x)<\mathcal R(x)\quad (x\neq x_{0}).
\end{equation}
\end{theorem}

\begin{proof}
We first construct the branch near $x_{0}$. Write
\begin{equation}\label{eq:polar-mass}
h=\mathcal R\cos\theta,\qquad h'=\mathcal R\sin\theta,
\end{equation}
and define
\begin{equation}\label{eq:rho-def}
\rho=\frac{\mathcal R'}{\mathcal R}=\frac{a e^{2ax}+b e^{2bx}}{e^{2ax}+e^{2bx}}.
\end{equation}
Away from $\theta=0$, differentiating the first relation in Eq.~\eqref{eq:polar-mass} and comparing it with the second gives
\begin{equation}\label{eq:theta-ode}
\theta'=\rho\cot\theta-1.
\end{equation}
At $x=x_{0}$ one has $\rho(x_{0})=0$ and $\rho'(x_{0})=2\kappa^{2}$. Setting $t=x-x_{0}$, write $\rho=t\eta(t)$, where $\eta$ is analytic and $\eta(0)=2\kappa^{2}$. With $\theta=t u(t)$, Eq.~\eqref{eq:theta-ode} becomes
\begin{equation}\label{eq:singular-u}
t u'=\frac{\eta(t)}{u}\bigl[t u\cot(tu)\bigr]-1-u.
\end{equation}
Since $z\cot z$ is analytic at the origin, any solution with a finite limit $u(0)$ must satisfy
\begin{equation}\label{eq:u0-roots}
u(0)^2+u(0)=2\kappa^2.
\end{equation}
The two roots are $c$ and $-1-c$. We first take the positive root $u(0)=c$. Setting $v=u-c$ gives
\begin{equation}\label{eq:v-equation}
tv'+\Lambda v=\mathcal G(t,v),\qquad
\Lambda=2+\frac{1}{c}>0,
\end{equation}
where $\mathcal G$ is analytic near $(0,0)$ and satisfies
$\mathcal G(0,0)=\mathcal G_v(0,0)=0$. Equivalently,
\begin{equation}\label{eq:v-integral}
v(t)=\int_{0}^{1}s^{\Lambda-1}\mathcal G(st,v(st))\,\mathrm{d}s.
\end{equation}

We spell out the local fixed-point step because the equation is singular at $t=0$. There are $r_0,R_0,C>0$ such that, for $|t|\le r_0$ and $|v|\le R_0$,
\begin{equation}\label{eq:G-local-bounds}
|\mathcal G(t,v)|
\le
C\bigl(|t|+|t||v|+|v|^2\bigr),
\qquad
|\mathcal G_v(t,v)|
\le
C\bigl(|t|+|v|\bigr).
\end{equation}
Choose $M>2C/\Lambda$ and then $r>0$ so small that
$r\le r_0$, $Mr\le R_0$, and
\begin{equation}
\frac{C(1+M)r}{\Lambda}<\frac12.
\end{equation}
On the closed ball
\begin{equation}
\mathcal B_{r,M}
=
\left\{
v\in H^\infty(D_r):
v(0)=0,\ 
\|v\|_\infty\le Mr
\right\},
\qquad
D_r=\{t\in\mathbb C:|t|<r\},
\end{equation}
Eq.~\eqref{eq:G-local-bounds} gives
\begin{equation}
\|\mathcal T v\|_\infty
\le
\frac{C}{\Lambda}
\left(r+Mr^2+M^2r^2\right)
\le Mr
\end{equation}
after decreasing $r$ if necessary. Moreover,
\begin{equation}
\|\mathcal T v-\mathcal T w\|_\infty
\le
\frac{C(1+M)r}{\Lambda}
\|v-w\|_\infty
<
\frac12\|v-w\|_\infty.
\end{equation}
Thus $\mathcal T$ is a contraction and determines a unique holomorphic fixed point; its restriction to the real axis is the desired real-analytic solution.

The same fixed point is unique among bounded real solutions with $v(t)\to0$. Indeed, any such solution of Eq.~\eqref{eq:v-equation}, $C^1$ for $t\neq0$, satisfies the integral equation \eqref{eq:v-integral}. For $t>0$, multiplying Eq.~\eqref{eq:v-equation} by $t^{\Lambda-1}$ and integrating from $0$ to $t$ gives
\begin{equation}
t^\Lambda v(t)
=
\int_0^t
\tau^{\Lambda-1}
\mathcal G(\tau,v(\tau))
\,\mathrm d\tau,
\end{equation}
because boundedness implies $t^\Lambda v(t)\to0$ as $t\to0^+$. Dividing by $t^\Lambda$ and setting $\tau=st$ yields Eq.~\eqref{eq:v-integral}. For $t<0$ the same argument applies to $\widetilde v(s)=v(-s)$ on $s>0$. After shrinking the interval, Eq.~\eqref{eq:v-integral} and $\mathcal G(t,0)=O(t)$ therefore give
\begin{equation}
m(t)
\le
C_1|t|+\frac12 m(t),
\qquad
m(t)=\sup_{|\tau|\le|t|}|v(\tau)|,
\end{equation}
and hence $v(t)=O(t)$. The solution then belongs to a ball of the form $\mathcal B_{r,M}$, where contraction uniqueness applies. Reconstructing $h=\mathcal R\cos\theta$ gives Eq.~\eqref{eq:center-data}, with $h'>0$ immediately to the right of $x_{0}$ and $h'<0$ immediately to the left.

We next continue this local solution over the full real axis. Let
\begin{equation}\label{eq:D-def}
D(x)=\mathcal R(x)^{2}-h(x)^{2}=h'(x)^{2}.
\end{equation}
To the right of $x_{0}$ we use $h'=\sqrt{D}$. If $D$ had a first zero at some finite $x_{*}>x_{0}$, then $D'(x_{*})\leq0$ because $D>0$ immediately to its left. At such a point $h'(x_{*})=0$, whereas
\begin{equation}\label{eq:D-boundary}
D'(x_{*})=2\mathcal R(x_{*})\mathcal R'(x_{*})>0,
\end{equation}
which is impossible. On the left we use $h'=-\sqrt{D}$; a first zero $x_{*}<x_{0}$ approached from the right would require $D'(x_{*})\geq0$, while $2\mathcal R\mathcal R'<0$ there. Thus $D$ remains positive for every finite $x\neq x_{0}$. On any compact interval, $|h|$ and $|h'|$ are bounded by $\mathcal R$, while the monotonicity just established gives $h\geq\mathcal R_{0}$. Hence neither a finite blow-up nor a zero of $h$ can occur, and standard continuation of the nonsingular square-root equation extends the solution to all $x\in\mathbb{R}$. It is analytic away from $x_{0}$, and the local construction already gives analyticity at $x_{0}$.

It remains to prove uniqueness in the stated class. Let $g>0$ be any $C^{2}(\mathbb{R})$ solution of Eq.~\eqref{eq:mass-equation}. Differentiating $g^{2}+g'^{2}=\mathcal R^{2}$ gives
\begin{equation}\label{eq:stationary-identity}
g'(g+g'')=\mathcal R\mathcal R'.
\end{equation}
Therefore any zero of $g'$ must occur at $x_{0}$. Such a zero must exist. If $g'>0$ everywhere, then $g$ is bounded as $x\to-\infty$, whereas $\mathcal R(x)\to\infty$; Eq.~\eqref{eq:mass-equation} then gives $g'\geq\mathcal R/2$ sufficiently far to the left. For a fixed $x_{1}$ in that region,
\begin{equation}\label{eq:left-contradiction}
0<g(x_{1})-g(x)=\int_{x}^{x_{1}}g'(y)\,\mathrm{d}y\geq\frac{1}{2}\int_{x}^{x_{1}}\mathcal R(y)\,\mathrm{d}y,
\end{equation}
and the right-hand side diverges as $x\to-\infty$, while the left-hand side is bounded above by $g(x_{1})$. This is a contradiction. The case $g'<0$ everywhere is excluded analogously as $x\to+\infty$. Hence $g'(x_{0})=0$, and positivity gives $g(x_{0})=\mathcal R_{0}$.

Write $g''(x_{0})=\gamma\mathcal R_{0}$. Since $(\mathcal R^{2})''(x_{0})=4\kappa^{2}\mathcal R_{0}^{2}$, comparison of the quadratic terms in $g^{2}+g'^{2}=\mathcal R^{2}$ yields
\begin{equation}\label{eq:gamma-equation}
\gamma^{2}+\gamma=2\kappa^{2}.
\end{equation}
Thus $\gamma=c$ or $\gamma=-1-c$. The second possibility is incompatible with global positivity. It gives $g'<0$ immediately to the right of $x_{0}$, and Eq.~\eqref{eq:stationary-identity} shows that $g'$ can never vanish again. Consequently $g(x)\leq\mathcal R_{0}$ for $x>x_{0}$, while $\mathcal R(x)\to\infty$; eventually Eq.~\eqref{eq:mass-equation} implies $g'\leq-\mathcal R/2$, and integration forces $g$ to cross zero at a finite value of $x$. Hence $\gamma=c$.

Since $g>0$ near $x_{0}$, we may introduce a continuous angle $\vartheta\in(-\pi/2,\pi/2)$ by
\begin{equation}\label{eq:candidate-angle}
g=\mathcal R\cos\vartheta,\qquad g'=\mathcal R\sin\vartheta.
\end{equation}
The expansion above implies $\vartheta(t)/t\to c$ as $t\to0$. Thus $u_g(t)=\vartheta(t)/t$ is bounded and tends to $c$, so the local uniqueness statement following from Eq.~\eqref{eq:v-integral} applies to $g$. It follows that $g$ coincides with the constructed branch in a neighbourhood of $x_{0}$. Ordinary uniqueness for the nonsingular square-root equation then extends this equality to both sides. This proves the theorem.
\end{proof}

By Eq.~\eqref{eq:mass-charge-relations}, $x_{0}$ is also the unique point on the global branch at which the scalar charge vanishes. Theorem~\ref{thm:global-branch} is a statement about the extremal mass function: positivity and $C^{2}$ regularity over the full open charge cone select a single analytic branch of Eq.~\eqref{eq:mass-equation}. No boundary condition at the single-charge limits $q\to0$ or $q\to\infty$ is required. The theorem does not assert uniqueness of arbitrary spacetime solutions, nor does it assume that the inner endpoint associated with every coupling is a smooth extremal horizon. That separate issue will be addressed in Sec.~\ref{Sec:Regularity}. With the mass branch fixed, we can now compare different charge ratios without any further branch choice.

\section{Interactions along the mass branch}\label{Sec:Interactions}

The sign of the leading long-range force between two extremal states is determined by the coefficient \cite{E2MD2026,Heidenreich2020,TodaForce2022,EMD2023}
\begin{equation}\label{eq:force-coefficient}
\mathcal F_{12}=Q_1Q_2+P_1P_2-M_1M_2-\Sigma_1\Sigma_2,
\end{equation}
where $\mathcal F_{12}>0$ corresponds to repulsion. Using Eq.~\eqref{eq:mass-charge-relations}, this becomes
\begin{equation}\label{eq:force-kernel-relation}
\mathcal F_{12}=-P_1P_2e^{-b(x_1+x_2)}\mathcal K_{12},
\end{equation}
with
\begin{equation}\label{eq:force-kernel}
\mathcal K_{12}=h_1h_2+h'_1h'_2-e^{a(x_1+x_2)}-e^{b(x_1+x_2)},
\end{equation}
where $h_i=h(x_i)$. If $x_1=x_2$, the two charge vectors lie on the same ray of the positive charge cone and Eq.~\eqref{eq:mass-equation} gives $\mathcal K_{12}=0$. The nontrivial problem is therefore to determine the sign for different charge ratios.

\subsection{The global force sign}\label{Sec:Force}

The force kernel admits a simple geometric interpretation. In addition to the mass angle $\theta$ introduced in Eq.~\eqref{eq:polar-mass}, define the bare-charge angle $\alpha$ by
\begin{equation}\label{eq:bare-angle}
e^{bx}=\mathcal R\cos\alpha,\qquad e^{ax}=\mathcal R\sin\alpha,\qquad \alpha=\arctan e^{\delta x}\in(0,\pi/2).
\end{equation}
Since $h>0$, the angle $\theta$ is globally defined in $(-\pi/2,\pi/2)$. Let
\begin{equation}\label{eq:angular-variables}
k=\alpha'=\delta\sin\alpha\cos\alpha>0,\qquad p=\frac{\mathrm d\theta}{\mathrm d\alpha}=\frac{\theta'}{k}.
\end{equation}
Here and below, a prime on $p$ denotes differentiation with respect to $x$. With $\rho=\mathcal R'/\mathcal R$ as in Eq.~\eqref{eq:rho-def},
\begin{equation}\label{eq:rho-k-identities}
\rho'=2k^2,\qquad k^2=(a-\rho)(\rho-b),\qquad \frac{k'}{k}=a+b-2\rho.
\end{equation}
At $x=x_0$, where $k=\kappa$ and $\theta'=c$,
\begin{equation}\label{eq:p-center}
p(x_0)=\frac{c}{\kappa},\qquad p(x_0)^2=\frac{2c}{c+1}.
\end{equation}
It follows that $p(x_0)<1$, $p(x_0)=1$, or $p(x_0)>1$ according as $ab>-1$, $ab=-1$, or $ab<-1$.

The angle equation \eqref{eq:theta-ode} may be written as
\begin{equation}\label{eq:p-angle-equation}
1+kp=\rho\cot\theta.
\end{equation}
Differentiating this relation and using Eq.~\eqref{eq:rho-k-identities} gives, at a possible crossing away from $x_0$,
\begin{equation}\label{eq:p-barriers}
\left.p'\right|_{p=0}=\frac{2k}{\rho},\qquad
\left.p'\right|_{p=1}=-\frac{1+ab}{\rho}.
\end{equation}
These two identities determine the global behavior of $p$. Since $p(x_0)>0$, a first zero to the right of $x_0$ would require $p'\leq0$, whereas Eq.~\eqref{eq:p-barriers} gives $p'>0$; a first zero to the left would require $p'\geq0$, whereas $p'<0$. Hence
\begin{equation}\label{eq:p-positive}
p(x)>0
\end{equation}
throughout the real axis.

Suppose first that $ab>-1$, so that $p(x_0)<1$. A first contact with $p=1$ on the right would require $p'\geq0$, while Eq.~\eqref{eq:p-barriers} gives $p'<0$. On the left the required and actual signs are again opposite. Thus $p<1$ everywhere. When $ab<-1$, the central value lies above one and the same argument, with all relevant signs reversed, gives $p>1$ everywhere. At the critical coupling $ab=-1$, with $b=-1/a$, the known exact positive branch \cite{E2MD2026} satisfies
\begin{equation}\label{eq:critical-orthogonal-map}
\begin{pmatrix}
h\\
h'
\end{pmatrix}
=
\frac{1}{\sqrt{1+a^2}}
\begin{pmatrix}
a&1\\
-1&a
\end{pmatrix}
\begin{pmatrix}
e^{-x/a}\\
e^{ax}
\end{pmatrix}.
\end{equation}
The mass vector is therefore a constant rotation of the bare-charge vector, and Theorem~\ref{thm:global-branch} identifies this solution with the global positive branch. Consequently,
\begin{equation}\label{eq:p-global}
0<p<1\quad(ab>-1),\qquad p=1\quad(ab=-1),\qquad p>1\quad(ab<-1).
\end{equation}

\begin{theorem}[Long-range force sign]\label{thm:force-sign}
For two states with $Q_i,P_i>0$ and finite charge ratios, $\mathcal F_{12}=0$ if the charge vectors are proportional or if $ab=-1$. If the charge vectors are nonproportional and $ab\neq-1$, then
\begin{equation}\label{eq:physical-force-sign}
\operatorname{sgn}\mathcal F_{12}=-\operatorname{sgn}(1+ab).
\end{equation}
\end{theorem}

\begin{proof}
Assume $x_1<x_2$ and define
\begin{equation}\label{eq:angle-differences}
A=\alpha_2-\alpha_1,\qquad T=\theta_2-\theta_1.
\end{equation}
The vectors $(h,h')$ and $(e^{bx},e^{ax})$ have the same length $\mathcal R$, so
\begin{equation}\label{eq:kernel-angle}
\frac{\mathcal K_{12}}{\mathcal R_1\mathcal R_2}=\cos T-\cos A.
\end{equation}
Here $0<A<\pi/2$, while $p>0$ and $\theta\in(-\pi/2,\pi/2)$ imply $0<T<\pi$. Moreover,
\begin{equation}\label{eq:T-integral}
T=\int_{\alpha_1}^{\alpha_2}p(\alpha)\,\mathrm d\alpha.
\end{equation}
Equation~\eqref{eq:p-global} gives $T<A$ for $ab>-1$, $T=A$ for $ab=-1$, and $T>A$ for $ab<-1$. Since $\cos z$ is strictly decreasing on $(0,\pi)$,
\begin{equation}\label{eq:kernel-sign}
\operatorname{sgn}\mathcal K_{12}=\operatorname{sgn}(1+ab)
\end{equation}
for $x_1\neq x_2$ away from criticality. Equation~\eqref{eq:force-kernel-relation} then proves Eq.~\eqref{eq:physical-force-sign}.
\end{proof}

Thus any two nonproportional positive charge vectors attract for $ab>-1$ and repel for $ab<-1$. The statement concerns the leading long-range interaction; it does not assert an all-distance two-body force law.

\subsection{Mass convexity and binding energy}\label{Sec:Binding}

The same transition is encoded in the geometry of the extremal mass surface. The relation between force, mass curvature, and binding energy was already observed in EMD theory \cite{EMD2023}; it is not universal once the matter sector is enlarged \cite{EMDA2025}. Define
\begin{equation}\label{eq:C-def}
\mathcal C=h''-(a+b)h'+ab\,h.
\end{equation}
Since
\begin{equation}
(\partial_x-2a)(\partial_x-2b)\mathcal R^2=0,\qquad \mathcal R^2=h^2+h'^2,
\end{equation}
direct expansion gives the identity
\begin{equation}\label{eq:C-identity}
0=h'\mathcal C'+\mathcal C^2+\bigl[(a+b)h'+(1-2ab)h\bigr]\mathcal C+(1+ab)(h'-ah)(h'-bh).
\end{equation}
No division by $h'$ is involved.

To determine the sign of the last factor, note that $\theta'=kp>0$. Equation~\eqref{eq:theta-ode} therefore gives
\begin{equation}\label{eq:w-rho}
\frac{h'}{h}=\tan\theta=\frac{\rho}{1+\theta'}.
\end{equation}
For $x>x_0$ one has $0<\rho<a$, while for $x<x_0$ one has $b<\rho<0$. Hence, by continuity also at $x_0$,
\begin{equation}\label{eq:log-slope-bound}
b<\frac{h'}{h}<a.
\end{equation}
In particular,
\begin{equation}\label{eq:factor-negative}
(h'-ah)(h'-bh)<0.
\end{equation}

At the center,
\begin{equation}\label{eq:C-center}
\frac{\mathcal C(x_0)}{\mathcal R_0}=c+ab=\frac{c(1-c)}{2},\qquad
1+ab=\frac{(1-c)(c+2)}{2},
\end{equation}
so $\mathcal C(x_0)$ has the sign of $1+ab$. At a zero of $\mathcal C$ away from $x_0$, Eq.~\eqref{eq:C-identity} reduces to
\begin{equation}\label{eq:C-crossing}
\mathcal C'=-\frac{(1+ab)(h'-ah)(h'-bh)}{h'}.
\end{equation}
If $1+ab>0$, a first zero to the right of $x_0$ would require $\mathcal C'\leq0$, whereas Eqs.~\eqref{eq:factor-negative} and \eqref{eq:C-crossing} give $\mathcal C'>0$; on the left the two signs are again incompatible. If $1+ab<0$, the same argument with reversed signs excludes a zero. At $ab=-1$, the critical solution \eqref{eq:critical-orthogonal-map} gives $\mathcal C\equiv0$. Therefore
\begin{equation}\label{eq:C-sign}
\operatorname{sgn}\mathcal C=\operatorname{sgn}(1+ab)
\end{equation}
throughout the finite charge-ratio branch away from criticality.

The quantity $\mathcal C$ has a direct interpretation in charge space. From $M=Pf(q)$, $q=e^{\delta x}$ and $f=e^{-bx}h$,
\begin{equation}\label{eq:f-second}
f''(q)=\frac{e^{-bx}}{\delta^2q^2}\,\mathcal C.
\end{equation}
The Hessian with respect to $(Q,P)$ is consequently
\begin{equation}\label{eq:mass-hessian}
\nabla^2M=\frac{f''(q)}{P}
\begin{pmatrix}
1&-q\\
-q&q^2
\end{pmatrix},
\qquad
\xi^{\mathsf T}\nabla^2M\,\xi=\frac{f''(q)}{P}(\xi_Q-q\xi_P)^2.
\end{equation}
Thus the extremal mass is convex for $ab>-1$ and concave for $ab<-1$. Away from criticality, the only null direction of the Hessian is the homogeneous scaling direction $(Q,P)$. At $ab=-1$ the mass becomes linear,
\begin{equation}\label{eq:critical-linear-mass}
M(Q,P)=\frac{Q+aP}{\sqrt{1+a^2}}.
\end{equation}

Following the convention of Ref.~\cite{EMD2023}, define the mass difference
\begin{equation}\label{eq:DeltaM-def}
\Delta M=M(Q_1+Q_2,P_1+P_2)-M(Q_1,P_1)-M(Q_2,P_2),
\end{equation}
which was referred to there as the binding energy. If one instead defines the binding energy with the opposite sign, all statements about its sign below are reversed. Let
\begin{equation}\label{eq:binding-variables}
P_T=P_1+P_2,\qquad \lambda=\frac{P_1}{P_T},\qquad q_i=\frac{Q_i}{P_i},\qquad \bar q=\lambda q_1+(1-\lambda)q_2.
\end{equation}
For $q_1<q_2$, Taylor's formula with integral remainder gives
\begin{equation}\label{eq:DeltaM-integral}
-\frac{\Delta M}{P_T}
=
\lambda\int_{q_1}^{\bar q}(u-q_1)f''(u)\,\mathrm du
+
(1-\lambda)\int_{\bar q}^{q_2}(q_2-u)f''(u)\,\mathrm du.
\end{equation}
The weights are strictly positive in the interiors of the two intervals. Equations~\eqref{eq:C-sign} and \eqref{eq:f-second} therefore imply, for nonproportional positive charges away from criticality,
\begin{equation}\label{eq:master-sign-relation}
\boxed{
\operatorname{sgn}\Delta M
=
\operatorname{sgn}\mathcal F_{12}
=
-\operatorname{sgn}(1+ab).
}
\end{equation}
The equality $\Delta M=0$ occurs precisely when the two charge vectors are proportional or when $ab=-1$.

The same conclusion holds for any finite collection of positive charge vectors. Writing
\begin{equation}\label{eq:finite-partition}
\Delta M
=
P_T\left[
f\!\left(\sum_i\lambda_iq_i\right)
-\sum_i\lambda_i f(q_i)
\right],
\qquad
P_T=\sum_iP_i,\qquad
\lambda_i=\frac{P_i}{P_T},
\end{equation}
strict convexity or concavity transverse to the homogeneous direction gives the same sign as in Eq.~\eqref{eq:master-sign-relation}, with equality away from criticality only when all $q_i$ coincide. This mass comparison is an energetic statement within the extremal family; by itself it does not determine a dynamical fragmentation or merger process, and entropy supplies a separate criterion in related black-hole systems \cite{Geng2019,CveticGibbonsLuPope2018}. The long-range force, the transverse curvature of the mass surface, and the mass difference for every positive-charge partition therefore change sign at the same coupling.

\section{Extensions and spacetime interpretation}\label{Sec:Extensions}

The sign relations derived above are properties of the extremal mass branch. Two further questions help delimit their physical scope. First, the known dimensional rescaling extends the argument to $D>4$ without repeating the proof. Second, in four dimensions the same mass branch can be reconstructed as a radial exterior, allowing us to distinguish finite-area extremal behavior from genuine smooth-horizon regularity.

\subsection{Higher dimensions}\label{Sec:HigherD}

For $D\geq4$, the E2MD mass equation takes the form \cite{E2MD2026}
\begin{equation}\label{eq:D-mass-equation}
\frac{(D-2)^2}{4}\left(\frac{\mathrm dh_D}{\mathrm dx}\right)^2+\frac{(D-2)(D-3)}{2}h_D^2=e^{2ax}+e^{2bx}.
\end{equation}
Introduce
\begin{equation}\label{eq:D-rescaling}
A_D^2=\frac{2}{(D-2)(D-3)},\qquad B_D^2=\frac{D-2}{2(D-3)},
\end{equation}
and set
\begin{equation}\label{eq:D-hatted}
h_D=A_D\widehat h,\qquad x=B_D\widehat x,\qquad \widehat a=B_Da,\qquad \widehat b=B_Db.
\end{equation}
Equation~\eqref{eq:D-mass-equation} then reduces to
\begin{equation}\label{eq:D-canonical}
\widehat h_{\widehat x}^{\,2}+\widehat h^2=e^{2\widehat a\widehat x}+e^{2\widehat b\widehat x},
\end{equation}
which is precisely the canonical equation studied in Secs.~\ref{Sec:MassBranch} and \ref{Sec:Interactions}. The physical charge ratio is unchanged,
\begin{equation}
q=e^{(\widehat a-\widehat b)\widehat x},
\end{equation}
and the dimensional normalization of both the mass and the leading force introduces only positive factors. The global sign relations therefore carry over directly:
\begin{equation}\label{eq:D-master-sign}
\boxed{
\operatorname{sgn}\mathcal F_{12}^{(D)}
=
\operatorname{sgn}\Delta M_D
=
-\operatorname{sgn}\left[
ab+\frac{2(D-3)}{D-2}
\right]
}
\end{equation}
for nonproportional positive charge vectors away from
\begin{equation}\label{eq:D-critical}
ab=-\frac{2(D-3)}{D-2}.
\end{equation}
The critical surface and the rescaling \eqref{eq:D-rescaling} are known \cite{E2MD2026}; what is added here is that the global force and mass-curvature arguments of Sec.~\ref{Sec:Interactions} apply to the rescaled equation without further assumptions. This dimensional map concerns the mass relation and does not, by itself, imply smoothness of the corresponding higher-dimensional extremal horizons.

\subsection{Exterior geometry and horizon regularity}\label{Sec:Regularity}

We now return to four dimensions. First-order formulations of static extremal black holes are well established \cite{Perz2009,Trigiante2012}; in the present setting they provide a direct reconstruction of the exterior associated with the global mass branch. This is an exterior construction: no regularity assumption at $r=0$ is imposed at this stage. Consider the isotropic ansatz
\begin{equation}\label{eq:radial-metric}
\mathrm ds^2=-e^{2U}\mathrm dt^2+e^{-2U}\left(\mathrm dr^2+r^2\mathrm d\Omega_2^2\right),\qquad \tau=\frac{1}{r},
\end{equation}
with $U,\phi\to0$ as $\tau\to0$. The Maxwell equations integrate to
\begin{equation}\label{eq:radial-maxwell}
F_{tr}=\frac{Q}{r^2}e^{2U+2a\phi},\qquad \widetilde F_{tr}=\frac{P}{r^2}e^{2U+2b\phi},
\end{equation}
and it is convenient to define
\begin{equation}\label{eq:effective-potential}
\mathcal V(\phi)=Q^2e^{2a\phi}+P^2e^{2b\phi}.
\end{equation}
With a dot denoting $\mathrm d/\mathrm d\tau$, the independent radial equations are \cite{E2MD2026}
\begin{equation}\label{eq:radial-second-order}
\ddot U=e^{2U}\mathcal V,\qquad
\ddot\phi=\frac{1}{2}e^{2U}\mathcal V_\phi,\qquad
\dot U^2+\dot\phi^2=e^{2U}\mathcal V.
\end{equation}

For fixed $Q,P>0$, let
\begin{equation}\label{eq:W-definition}
x=\frac{1}{a-b}\log\frac{Q}{P},\qquad A=Pe^{-bx},\qquad W(\phi)=A\,h(x+\phi).
\end{equation}
The mass equation immediately gives
\begin{equation}\label{eq:W-identity}
W^2+W_\phi^2=\mathcal V.
\end{equation}
Consequently,
\begin{equation}\label{eq:first-order-flow}
\dot U=-e^UW,\qquad \dot\phi=-e^UW_\phi,\qquad U(0)=\phi(0)=0
\end{equation}
solves Eq.~\eqref{eq:radial-second-order}. Differentiating the two first-order equations and using Eq.~\eqref{eq:W-identity} reproduces the second-order system, while the Hamiltonian constraint follows immediately. At spatial infinity,
\begin{equation}\label{eq:flow-asymptotic-charges}
M=W(0)=Pe^{-bx}h(x),\qquad
\Sigma=-W_\phi(0)=-Pe^{-bx}h'(x),
\end{equation}
in agreement with Eq.~\eqref{eq:mass-charge-relations}.

The global behavior of the flow is particularly simple. Set
\begin{equation}
Y=e^{-U}.
\end{equation}
Then
\begin{equation}\label{eq:Y-flow}
\dot Y=W,\qquad \dot\phi=-\frac{W_\phi}{Y}.
\end{equation}
Since $h$ has a unique minimum at $x_0$, the function $W$ has a unique minimum at
\begin{equation}\label{eq:phi-star}
\phi_\star=x_0-x,\qquad W_\star=A\mathcal R_0>0.
\end{equation}
Along the flow,
\begin{equation}\label{eq:W-monotonic}
\dot W=-\frac{W_\phi^2}{Y}\leq0,
\end{equation}
and therefore
\begin{equation}\label{eq:Y-bounds}
1+W_\star\tau\leq Y(\tau)\leq1+M\tau.
\end{equation}
The sign of $W_\phi$ on either side of its unique minimum shows that $\phi$ moves monotonically toward $\phi_\star$ and cannot cross it. Hence $Y$ remains positive and $\phi$ remains finite on every finite interval of $\tau$, so the reconstructed exterior is smooth for all $r>0$.

To determine the inner limit, set $s=\log Y$. Equation~\eqref{eq:Y-flow} gives
\begin{equation}\label{eq:phi-s-flow}
\frac{\mathrm d\phi}{\mathrm ds}=-\frac{W_\phi}{W}.
\end{equation}
Since $Y\geq1+W_\star\tau$, one has $s\to\infty$ as $\tau\to\infty$. Monotonicity then implies
\begin{equation}
\phi\longrightarrow\phi_\star,\qquad W\longrightarrow W_\star.
\end{equation}
It follows from Eq.~\eqref{eq:Y-flow} that
\begin{equation}
\frac{Y}{\tau}\longrightarrow W_\star.
\end{equation}
The areal radius $\mathscr R=rY$ therefore approaches
\begin{equation}\label{eq:inner-area}
\mathscr R\longrightarrow W_\star,\qquad
\mathcal A_{\rm inner}\longrightarrow4\pi W_\star^2,
\end{equation}
while
\begin{equation}\label{eq:inner-gtt}
g_{tt}=-Y^{-2}\sim-\frac{r^2}{W_\star^2}.
\end{equation}
Thus the global positive mass branch reconstructs an asymptotically flat exterior with a finite-area inner endpoint and the characteristic quadratic redshift of an extremal geometry. Neither property, however, guarantees a smooth horizon extension.

The distinction is already visible in the scalar field. Write
\begin{equation}
\varepsilon=\phi-\phi_\star.
\end{equation}
Since
\begin{equation}
\frac{W_{\phi\phi}(\phi_\star)}{W_\star}
=
\frac{h''(x_0)}{h(x_0)}
=c,
\end{equation}
Eq.~\eqref{eq:phi-s-flow} gives
\begin{equation}\label{eq:epsilon-flow}
\frac{\mathrm d\varepsilon}{\mathrm ds}=-c\varepsilon+O(\varepsilon^2).
\end{equation}
For a nonconstant flow, the linear term first gives exponential decay, after which the nonlinear correction is integrable. Hence $e^{cs}\varepsilon$ approaches a finite nonzero constant. Using $Y/\tau\to W_\star$, one obtains
\begin{equation}\label{eq:near-endpoint}
\varepsilon=C_r r^c[1+o(1)],\qquad
\phi_r=cC_r r^{c-1}[1+o(1)],\qquad C_r\neq0.
\end{equation}
The derivative asymptotic follows directly from the first-order flow.

The regularity of the endpoint can be tested in a freely propagated frame. Introduce the ingoing coordinate
\begin{equation}
v=t+\int Y^2\,\mathrm dr,
\end{equation}
for which
\begin{equation}\label{eq:ingoing-metric}
\mathrm ds^2=-Y^{-2}\mathrm dv^2+2\,\mathrm dv\,\mathrm dr+\mathscr R(r)^2\mathrm d\Omega_2^2,\qquad \mathscr R=rY.
\end{equation}
The vector $l=\partial_r$ is an affinely parametrized radial null vector, while $E=\mathscr R^{-1}\partial_\vartheta$ is parallelly propagated along it. In these coordinates the electric fields have only $vr$ components, so
\begin{equation}
F_{r\lambda}F_r{}^\lambda=
\widetilde F_{r\lambda}\widetilde F_r{}^\lambda=0.
\end{equation}
Since $g_{rr}=0$, the $rr$ Einstein equation reduces to
\begin{equation}\label{eq:Rrr}
R_{rr}=-2\frac{\mathscr R_{rr}}{\mathscr R}=2\phi_r^2.
\end{equation}
The corresponding parallelly propagated tidal component is therefore
\begin{equation}\label{eq:tidal-curvature}
R_{lElE}=-\frac{\mathscr R_{rr}}{\mathscr R}
=\phi_r^2
\sim c^2C_r^2r^{2c-2}.
\end{equation}

Under the standing assumption $a>0>b$, the attractive region is
\begin{equation}\label{eq:attractive-c}
-1<ab<0\qquad\Longleftrightarrow\qquad0<c<1.
\end{equation}
Equation~\eqref{eq:tidal-curvature} then diverges at the inner endpoint for every nonconstant-scalar flow. The charge ratio $q=q_0$ is exceptional: in that case $\phi_\star=0$, $W_\phi(0)=0$, and uniqueness of the first-order flow gives $\phi\equiv0$. The solution then reduces to the corresponding constant-scalar extremal Reissner--Nordstr\"om-type geometry. Thus the attractive part of the global mass diagram cannot, in general, be interpreted as a family of smooth nonconstant-scalar extremal horizons.

At the critical coupling $ab=-1$, for which $c=1$, the tidal component in Eq.~\eqref{eq:tidal-curvature} remains finite for a nonconstant flow. For $c>1$ it vanishes at the endpoint, but this is still insufficient to establish smooth extendibility. Since $r$ is an affine parameter along $l$, the noninteger leading power in Eq.~\eqref{eq:near-endpoint} obstructs a smooth scalar extension whenever $C_r\neq0$. A necessary condition for an analytic nonconstant-scalar endpoint is therefore
\begin{equation}\label{eq:analytic-horizon-condition}
c=n\in\mathbb N_{>0},\qquad
-ab=\frac{n(n+1)}{2}.
\end{equation}
On the symmetric subfamily $b=-a$, this becomes
\begin{equation}\label{eq:triangular-condition}
a^2=\frac{n(n+1)}{2},
\end{equation}
which is precisely the discrete ``triangular'' condition obtained from horizon analyticity in dyonic Einstein--Maxwell--dilaton theory \cite{Poletti1995,Triangular2015}. The sign in Eq.~\eqref{eq:analytic-horizon-condition} reflects our convention $a>0>b$; on the symmetric branch the invariant statement is the positive condition \eqref{eq:triangular-condition}. Equation~\eqref{eq:analytic-horizon-condition} is necessary, not sufficient: it removes the fractional-power obstruction but does not by itself prove the existence of a smooth or analytic horizon. This distinction is consistent with the broader observation that smooth extremal horizons can require nongeneric restrictions on couplings and boundary data \cite{SmoothHorizons2025,AxionSingular2026}.

The force and mass-curvature results of Sec.~\ref{Sec:Interactions} therefore apply to the full global positive mass branch and to the finite-area exterior reconstructed above. Requiring a smooth horizon restricts the subset of that branch that admits an interpretation as smooth extremal black holes, but it does not alter the sign theorems themselves.

\section{Conclusions}\label{Sec:Conclusions}

We have shown that positivity over the full positive charge cone selects a unique global analytic branch of the extremal E2MD mass equation, with no integrability assumption and no boundary condition imposed at the single-charge limits. On this branch the attractive--repulsive transition is genuinely global in charge ratio: for nonproportional positive charges, the leading force, the transverse curvature of the extremal mass surface, and the mass difference under charge addition all change sign at $ab=-1$. The force--mass correspondence previously observed in EMD theory \cite{EMD2023} is therefore exact throughout this E2MD sector, while the contrasting EMDA behavior \cite{EMDA2025} shows that it is not a generic feature of scalar-coupled theories.

The dimensional rescaling carries the sign theorem directly to higher dimensions. The four-dimensional radial reconstruction also separates two notions that need not coincide: a global extremal mass branch and a family of smooth extremal black holes. Every point of the branch defines an asymptotically flat exterior with a finite-area inner endpoint, but smooth continuation imposes additional conditions. In particular, generic nonconstant-scalar flows in the attractive region develop divergent parallelly propagated tidal curvature. The interaction theorem therefore belongs to the full positive mass branch, whereas smooth-horizon regularity selects a more restricted subset.

Two questions seem especially worth pursuing. One is whether the discrete regularity condition identified here becomes sufficient under suitable assumptions, with the exact Toda solutions providing natural benchmarks \cite{Rank2Toda2026}. The other is how much of the force--mass-geometry correspondence survives additional matter fields or higher-derivative interactions; the EMDA example \cite{EMDA2025} and known higher-derivative corrections to extremal forces \cite{HigherDerivatives2022} already suggest nontrivial departures. The central result of the present analysis is therefore the global nature of the E2MD attractive--repulsive transition, while smooth-horizon regularity remains a separate and more restrictive question.

\section*{Statements and Declarations}

\paragraph{Funding}
No funds, grants, or other support were received.

\paragraph{Competing interests}
The author has no relevant financial or non-financial interests to disclose.

\paragraph{Author contributions}
Ye Zhou: Conceptualization, Methodology, Formal analysis, Investigation, Writing -- original draft, Writing -- review and editing.

\paragraph{Data availability}
No datasets were generated or analyzed during this study.

\paragraph{Use of generative AI}
The author used OpenAI ChatGPT (GPT-5.6 Sol, accessed September 2026) for literature organization, language editing, and auxiliary checks of selected calculations. All derivations, mathematical arguments, scientific interpretations, and literature claims were independently verified by the author, who takes full responsibility for the manuscript.

\bibliography{reference}

\end{document}